\documentclass[10pt]{ieeeconf}

\usepackage[bibstyle=ieee,giveninits=true]{biblatex}
\usepackage{url}
\usepackage[breaklinks]{hyperref}
\usepackage{breakurl}

\usepackage{amsmath,amssymb}
\usepackage{graphicx,graphicx,verbatim}
\usepackage{mathrsfs}
\usepackage[dvipsnames]{xcolor}
\usepackage{enumerate}
\usepackage[margin=0.75in]{geometry}
\usepackage{mathtools}
\usepackage[font=it]{caption}
\usepackage{subcaption}
\usepackage[outdir=./]{epstopdf}

\graphicspath{{./fig}}

\newtheorem{theorem}{Theorem}

\newtheorem{lemma}{Lemma}
\newtheorem{exx}{Example}
\newtheorem{assumption}{Assumption}

\renewcommand{\Re}{\ensuremath{\mathbb{R}}}

\title{Concurrent learning in high-order tuners for parameter identification}
\author{Justin H. Le and Andrew R. Teel\thanks{The authors are with the Electrical and Computer Engineering Department, University of California, Santa Barbara, USA. Email: {\em justinle@ucsb.edu}, {\em teel@ucsb.edu}. Research supported in part by the Air Force Office of Scientific Research under grant AFOSR FA9550-21-1-0452.}}

\begin{document}

\maketitle
\thispagestyle{empty}
\pagestyle{empty}

\begin{abstract}
High-order tuners are algorithms that show promise in achieving greater efficiency than classic gradient-based algorithms in identifying the parameters of parametric models and/or in facilitating the progress of a control or optimization algorithm whose adaptive behavior relies on such models. For high-order tuners, robust stability properties, namely uniform global asymptotic (and exponential) stability, currently rely on a persistent excitation (PE) condition. In this work, we establish such stability properties with a novel analysis based on a Matrosov theorem and then show that the PE requirement can be relaxed via a concurrent learning technique driven by sampled data points that are sufficiently rich. We show numerically that concurrent learning may greatly improve efficiency. We incorporate reset methods that preserve the stability guarantees while providing additional improvements that may be relevant in applications that demand highly accurate parameter estimates at relatively low additional cost in computation.
\end{abstract}

\section{Introduction}

The problem of identifying the parameters in a linear parametric model through the use of online measurements of input-output data has been studied extensively from the standpoint of the continuous-time gradient algorithm, which plays an important role in adaptive control \cite{narendra-2012}, \cite{tao-2003} and has found broad applications in areas such as optimal nonlinear control \cite{vamvoudakis-2010}, extremum seeking  \cite{poveda-benosman-2021}, and the analysis of machine learning algorithms \cite{gaudio-2019}. The gradient algorithm has been shown to achieve uniform global asymptotic (and exponential) stability of the desired parameter value, under a condition of persistent excitation which has been shown to be both necessary and sufficient \cite{Morgan77a}. Here, the distinction between uniform and non-uniform asymptotic stability has crucial implications in practice: unlike the non-uniform notion of stability, the uniform notion ensures robustness, in the sense of achieving ``total stability'' in the presence of bounded additive disturbances (\cite[Sec. 1B]{panteley-2001}) which are widely encountered in applications. 

Although persistent excitation characterizes uniform asymptotic stability for the gradient algorithm, the incorporation of ``concurrent learning'' has been shown to enable uniform asymptotic stability in the absence of persistent excitation \cite{chowdhary-2010}. Concurrent learning involves an augmentation of the gradient, using discretely sampled data measurements that collectively satisfy a condition of ``sufficient richness.'' Advantageously, sufficient richness can be characterized in terms of a rank condition on a matrix constructed from the data samples, a condition which can be computationally much simpler to verify in practice than the condition of persistent excitation.

In a separate thread of research, recent studies have shown that filtering the gradient can improve the performance of the gradient algorithm, sometimes significantly and in the absence of persistent excitation \cite{gaudio-dissertation}. Namely, the use of filtering can improve the convergence rate of the parameter estimate toward the desired parameter value or, in adaptive control problems, the convergence rate of the tracking error toward zero. The filtering procedures considered in \cite{gaudio-dissertation} and \cite{gaudio-2021} give rise to algorithms referred to as high-order tuners, which take inspiration from algorithms introduced in \cite{morse-1992}. In addition to offering improved performance in continuous-time settings, the recently developed high-order tuners serve as a basis for deriving novel discrete-time algorithms for various problems of estimation and learning with online data, with guarantees of efficiency in the sense of Nesterov's method for convex optimization \cite{moreu-2021} and of robustness in noisy and adversarial environments \cite{mcdonald-2021}, \cite{gaudio-2021b}. In continuous time, under a persistent excitation condition, uniform asymptotic stability properties of high-order tuners can be established using the method of analysis in \cite[Sec. 4.6]{annaswamy-2022}. However, the use of concurrent learning in high-order tuners has not yet been explored.

In this work, we establish uniform global asymptotic stability (UGAS) properties of continuous-time high-order tuners for parameter identification under two different conditions: $(1)$ persistent excitation and $(2)$ concurrent learning with sufficiently rich data. In Section \ref{sec:PE}, UGAS is established under persistent excitation, using an approach that we claim to be simpler than that of \cite[Sec. 4.6]{annaswamy-2022}. Whereas \cite{annaswamy-2022} shows uniform convergence by carefully examining solutions of a differential equation, we instead take advantage of a Matrosov theorem \cite{Loria05a}, which can be regarded as an analogue of the LaSalle invariance principle in the context of time-varying systems, with which uniform convergence is shown by combining infinitesimal conditions on Lyapunov-like functions together with observability-like conditions. In Section \ref{sec:CL}, we propose implementations of concurrent learning for high-order tuners in order to preserve UGAS in the absence of persistent excitation, given sufficiently rich data. We show that the resulting systems admit strict Lyapunov functions. In Section \ref{sec:softreset}, we propose the use of a technique inspired by reset methods in control and optimization (\cite{le-2021}, \cite{le-2021b}), which shows promise in improving the efficiency of high-order tuners that make use of concurrent learning. In Section \ref{sec:numerical}, numerical results show that concurrent learning can offer significant improvements in the convergence rate of high-order tuners for a parameter identification problem involving a regressor constructed from sinusoids.

\section{Notation and definitions}
For $x \in \Re^{n}$, we use $|x|:= \sqrt{x^{T} x}$. Given a pair $(x,u) \in \Re^{n} \times \Re^{m}$, by abuse of notation we sometimes consider $z:=(x,u)$ to be a vector in $\Re^{n+m}$. By the same abuse of notation, we sometimes write a function defined on $\Re^{n} \times \Re^{m}$ as a function defined on $\Re^{n + m}$, i.e.,  $f(x,u)$ may be written as $f(z)$ with $z=(x,u)$. A map $\varphi: \Re^{n} \times \Re_{\geq 0} \to \Re^{m}$ is said to be locally bounded in $x$ uniformly in $t$ if there exists a number $M > 0$ not dependent on $t$ such that $|\varphi(x, t)| \leq M$ for all $t$ and for all $x$ within a sufficiently small ball centered at the origin.

For a locally bounded function $f: \Re^{n} \times \Re_{\geq 0} \to \Re^{n}$ such that $x \mapsto f(x, t)$ is continuous uniformly in $t$ and $t \mapsto f(x, t)$ is piecewise continuous, the origin of the system $\dot{x} = f(x, t)$ is said to be uniformly globally stable (UGS) if there exists a class-$\mathcal{K}_{\infty}$ function $\gamma$ such that, for each initial condition $(x_{\circ}, t_{\circ}) \in \Re^{n} \times \Re_{\geq 0}$, each solution $x(\cdot)$ satisfies $|x(t)| \leq \gamma(|x_{\circ}|)$ for all $t \geq t_{\circ}$. The origin is said to be uniformly globally attractive (UGA) if for each $r > 0$ and $\sigma > 0$ there exists $T > 0$ such that, if the initial condition $(x_{\circ}, t_{\circ}) \in \Re^{n} \times \Re_{\geq 0}$ satisfies $|x_{\circ}| \leq r$, $|x(t)| \leq \sigma$ for all $t \geq t_{\circ} + T$. The origin is said to be uniformly globally asymptotically stable (UGAS) if it is UGS and UGA. The origin is said to be uniformly globally exponentially stable (UGES) if there exist $c > 0$ and $\alpha > 0$ such that, for each initial condition $(x_{\circ}, t_{\circ}) \in \Re^{n} \times \Re_{\geq 0}$, $|x(t)| \leq c|x_{\circ}| \exp(-\alpha(t - t_{0}))$ for all $t \geq t_{\circ}$.

\section{Uniform global asymptotic stability in high-order tuners via persistent excitation}
\label{sec:PE}

Suppose $y^{*}(t) = \phi^{T}(t)\theta^{*}$ for all $t \geq 0$, where $y^{*}: \mathbb{R}_{\geq 0} \rightarrow \mathbb{R}$ and $\phi: \mathbb{R}_{\geq 0} \rightarrow \mathbb{R}^{n}$ are known functions of time, and we wish to solve for $\theta^{*} \in \mathbb{R}^{n}$ online under the following assumption.  
\begin{assumption}
\label{assumption_PE}
The regressor $\phi(\cdot)$ is piecewise continuous, bounded, and persistently exciting. That is, there exist $M>0$, $T>0$, and $\delta>0$ such that $|\phi(t)| \leq M$ for all $t \geq 0$, and
$$
\int_{t}^{t+T} \phi(s) \phi^{T}(s) ds \geq \delta I \quad \forall t \geq 0.
$$
\end{assumption}

As a means of determining $\theta^{*}$, we follow the algorithmic development of \cite{gaudio-2021}, \cite[Ch. 5]{gaudio-dissertation}. Let $\theta \in \Re^{n}$, $y(t) \coloneqq \phi^{T}(t)\theta$, $\tilde{\theta} \coloneqq \theta - \theta^{*}$, $e_{y}(t) \coloneqq y(t) - y^{*} = \phi^{T}(t)\tilde{\theta}$, and $L_{t}(\theta) \coloneqq (1/2)\tilde{\theta}^{T}\phi(t)\phi^{T}(t)\tilde{\theta}$, so that $\nabla_{\theta} L_{t}(\theta) = \phi(t) e_{y}(t)$, and consider the following differential equation in the variable $x \coloneqq (\theta, \vartheta)$:
\begin{align}
\dot{x} &= f(x, t) \coloneqq \left[\begin{array}{c}-\beta(\theta - \vartheta)\mathcal{N}_{t} \\ -\gamma \nabla_{\theta} L_{t}(\theta)\end{array}\right], \label{eq:ht}
\end{align}
where $\beta \in \mathbb{R}_{> 0}$, $\gamma \in \mathbb{R}_{> 0}$, and $t \mapsto \mathcal{N}_{t}$ are to be selected for the purpose of achieving desired convergence properties for solutions of \eqref{eq:ht}. Choosing $\mathcal{N}_{t}$ to be dependent on $\phi(t)$ will be crucial for establishing stability properties of \eqref{eq:ht}. We focus on the case of
\begin{align}
\label{eq:Nt}
    \mathcal{N}_{t} \coloneqq 1 + \mu\phi^{T}(t)\phi(t) \quad \forall t \geq 0, \quad \mu \in \Re_{>0},
\end{align}
although other choices may be feasible (see \cite[Sec. 3]{ochoa-2021} for an example to consider). Equation \eqref{eq:ht} is referred to as a high-order tuner and is identical to \cite[Eq. 6]{gaudio-2021} and \cite[Eq. 5.6]{gaudio-dissertation}.

Another high-order tuner of interest is given as follows:
\begin{align}
\dot{x} &= f(x, t) \coloneqq \left[\begin{array}{c}-\beta(\theta - \vartheta) \\ \displaystyle{-\frac{\gamma}{\mathcal{N}_{t}}} \nabla_{\theta} L_{t}(\theta)\end{array}\right]. \label{eq:ht_normalized}
\end{align}
Equation \eqref{eq:ht_normalized} is identical to \cite[Eq. 6']{gaudio-2021} and \cite[Eq. 5.6']{gaudio-dissertation}.

For brevity, the dependence of $\phi$ and $e_{y}$ on $t$ will be suppressed hereafter.

\begin{theorem}
\label{thm:ugas_ht}
Under Assumption \ref{assumption_PE}, with $\mathcal{N}_{t}$ given by \eqref{eq:Nt}, if $\beta \geq 2\gamma/\mu$, the point $(\theta^{*}, \theta^{*})$ is uniformly globally asymptotically stable for \eqref{eq:ht}.
\end{theorem}
\begin{proof}
Recalling that $\tilde{\theta} = \theta - \theta^{*}$, let $p \coloneqq \vartheta - \theta$ and $\tilde{x} \coloneqq (\tilde{\theta}, p)$. We will use Matrosov's theorem \cite{Loria05a} to establish that the origin is UGAS for the system
\begin{align}
\label{eq:ht_tilde}
\dot{\tilde{x}} &= \tilde{f}(\tilde{x}, t) \coloneqq \left[\begin{array}{c}\beta \mathcal{N}_{t}p \\ -\beta \mathcal{N}_{t}p -\gamma \nabla_{\theta} L_{t}\left(\tilde{\theta} + \theta^{*}\right)\end{array}\right],
\end{align}
which will imply that $(\theta^{*}, \theta^{*})$ is UGAS for \eqref{eq:ht}. To begin, consider the Lyapunov function candidate
\begin{align}
V_{0}(\tilde{x}) &\coloneqq \frac{1}{\gamma}\left|\tilde{\theta} + p\right|^{2} + \frac{1}{\gamma}|p|^{2}. \label{eq:V0}
\end{align}
which is radially unbounded, positive definite, and continuously differentiable. For all $(\tilde{x}, t) \in \mathbb{R}^{2n} \times \mathbb{R}_{\geq 0}$, we have
\begin{align*}
&\left\langle\nabla V_{0}(\tilde{x}), \tilde{f}(\tilde{x}, t)\right\rangle = \frac{2}{\gamma}\Bigg[-\left\langle\tilde{\theta} + p, \; \gamma \nabla_{\theta} L_{t}\left(\tilde{\theta} + \theta^{*}\right)\right\rangle \\
&\qquad- \left\langle p, \; \beta \mathcal{N}_{t}p + \gamma \nabla_{\theta} L_{t}\left(\tilde{\theta} + \theta^{*}\right)\right\rangle \Bigg] \\
&= -2\left\langle\tilde{\theta}, \nabla_{\theta} L_{t}\left(\tilde{\theta} + \theta^{*}\right)\right\rangle - \frac{2\beta\mathcal{N}_{t}}{\gamma}|p|^{2} \\
&\qquad- 4\left\langle p, \; \nabla_{\theta} L_{t}\left(\tilde{\theta} + \theta^{*}\right)\right\rangle.
\end{align*}
Substituting $\nabla_{\theta} L_{t}\left(\tilde{\theta} + \theta^{*}\right) = \phi\phi^{T}\tilde{\theta}$ and $e_{y} = \phi^{T}\tilde{\theta}$, followed by $\mathcal{N}_{t} = 1 + \mu\phi^{T}\phi$, we have
\begin{align*}
&\left\langle\nabla V_{0}(\tilde{x}), \tilde{f}(\tilde{x}, t)\right\rangle \leq -2\left|e_{y}\right|^{2} - \frac{2\beta\mathcal{N}_{t}}{\gamma}|p|^{2} + 4|p||\phi||e_{y}| \\
&\qquad = -2\left|e_{y}\right|^{2} - \frac{2\beta}{\gamma}|p|^{2} - \frac{2\beta\mu}{\gamma}|\phi|^{2}|p|^{2} + 4|p||\phi||e_{y}|.
\end{align*}
With $\beta \geq 2\gamma/\mu$, it follows that
\begin{align}
&\left\langle\nabla V_{0}(\tilde{x}), \tilde{f}(\tilde{x}, t)\right\rangle \nonumber \\
&\qquad \leq -2\left|e_{y}\right|^{2} - \frac{2\beta}{\gamma}|p|^{2} - 4|\phi|^{2}|p|^{2} + 4|p||\phi||e_{y}| \nonumber \\
&\qquad = -\frac{2\beta}{\gamma}|p|^{2} - |e_{y}|^{2} - \left[|e_{y}| - 2|p||\phi|\right]^{2} \nonumber \\
&\qquad\leq -\frac{2\beta}{\gamma}|p|^{2} - |e_{y}|^{2} \eqqcolon Y_{0}(\tilde{x}, e_{y}) \leq 0, \label{eq:Y0}
\end{align}
and hence the origin of \eqref{eq:ht_tilde} is uniformly globally stable.
Next, we establish uniform global attractivity by building Matrosov functions as follows. Let
\begin{align}
\label{eq:V1}
V_{1}(\tilde{x}, t) &\coloneqq -\tilde{\theta}^{T}\left(\int_{t}^{\infty} \exp(t - \tau) \phi(\tau) \phi^{T}(\tau) d\tau\right) \tilde{\theta},
\end{align}
and note that Assumption \ref{assumption_PE} implies
\begin{align*}
V_{1}(\tilde{x}, t) &\leq -\exp(-T) \delta \tilde{\theta}^{T} \tilde{\theta} \quad \forall(\tilde{x}, t) \in \mathbb{R}^{2 n} \times \mathbb{R}_{\geq 0}.
\end{align*}
Then, for all $(\tilde{x}, t) \in \mathbb{R}^{2n} \times \mathbb{R}_{\geq 0}$, it holds that
\begin{align}
&\frac{\partial V_{1}(\tilde{x}, t)}{\partial t} + \frac{\partial V_{1}(\tilde{x}, t)}{\partial \tilde{x}} f(\tilde{x}, t) \nonumber \\
&\leq V_{1}(\tilde{x}, t) + \tilde{\theta}^{T} \phi \phi^{T} \tilde{\theta} + \beta M^{2}(1 + \mu M^{2})|\tilde{\theta}||p| \nonumber \\
&= V_{1}(\tilde{x}, t) + |e_{y}|^{2} + \beta M^{2}(1 + \mu M^{2})|\tilde{\theta}||p| \nonumber \\
&\leq -\exp(-T) \delta \tilde{\theta}^{T} \tilde{\theta} + |e_{y}|^{2} + \beta M^{2}(1 + \mu M^{2})|\tilde{\theta}||p| \label{eq:Y1} \\
&\eqqcolon Y_{1}(\tilde{x}, e_{y}). \nonumber
\end{align}
Note that $Y_{0}(\tilde{x}, e_{y}) = 0$ implies $p = 0$ and $e_{y} = 0$, which implies $Y_{1}(\tilde{x}, e_{y}) = -\exp(-T) \delta \tilde{\theta}^{T} \tilde{\theta} \leq 0$. Also note that $Y_{0}(\tilde{x}, e_{y}) = Y_{1}(\tilde{x}, e_{y}) = 0$ implies $p = 0$ and $\tilde{\theta} = 0$. Finally, note that $V_{0}$ is time-invariant, the maps $(\tilde{x}, t) \mapsto V_{1}(\tilde{x}, t)$ and $(\tilde{x}, t) \mapsto \phi^{T}(t)\tilde{\theta}$ are each locally bounded in $\tilde{x}$ uniformly in $t$, and both $Y_{0}$ and $Y_{1}$ are continuous. Thus, the conditions of Matrosov's theorem \cite{Loria05a} hold, and we conclude that the origin of \eqref{eq:ht_tilde} is UGAS.
\end{proof}

\begin{theorem}
\label{thm:ugas_ht_normalized}
Under Assumption \ref{assumption_PE}, with $\mathcal{N}_{t}$ given by \eqref{eq:Nt}, if $\beta \geq 2\gamma/\mu$, the point $(\theta^{*}, \theta^{*})$ is uniformly globally asymptotically stable for \eqref{eq:ht_normalized}.
\end{theorem}
\begin{proof}
Recalling that $\tilde{\theta} = \theta - \theta^{*}$, let $p \coloneqq \vartheta - \theta$ and $\tilde{x} \coloneqq (\tilde{\theta}, p)$. Reusing notation from the proof of Theorem \ref{thm:ugas_ht}, we will use Matrosov's theorem \cite{Loria05a} to establish that the origin is UGAS for the system
\begin{align}
\label{eq:ht_normalized_tilde}
\dot{\tilde{x}} &= \tilde{f}(\tilde{x}, t) \coloneqq \left[\begin{array}{c}\beta p \\ \displaystyle{-\beta p-\frac{\gamma}{\mathcal{N}_{t}}} \nabla_{\theta} L_{t}\left(\tilde{\theta} + \theta^{*}\right)\end{array}\right],
\end{align}
which will imply that $(\theta^{*}, \theta^{*})$ is UGAS for \eqref{eq:ht_normalized}. To begin, consider the Lyapunov function candidate \eqref{eq:V0}, which is radially unbounded, positive definite, and continuously differentiable. Observing that the right-hand side of \eqref{eq:ht_normalized_tilde} can be obtained by multiplying the right-hand side of \eqref{eq:ht_tilde} by $1/\mathcal{N}_{t}$, we have from \eqref{eq:Y0} that, for all $(\tilde{x}, t) \in \mathbb{R}^{2n} \times \mathbb{R}_{\geq 0}$ and for $\beta \geq 2\gamma/\mu$,
\begin{align}
&\left\langle\nabla V_{0}(\tilde{x}), \tilde{f}(\tilde{x}, t)\right\rangle \nonumber \\
&\qquad\leq \frac{1}{\mathcal{N}_{t}} \left\{-\frac{2\beta}{\gamma}|p|^{2} - \left|e_{y}\right|^{2}\right\} \label{eq:Y0_normalized} \\
&\qquad= \frac{1}{1 + \mu\phi^{T}\phi} \left\{-\frac{2\beta}{\gamma}|p|^{2} - \tilde{\theta}^{T}\phi\phi^{T}\tilde{\theta}\right\} \nonumber \\
&\qquad\eqqcolon Y_{0}(\tilde{x}, \phi) \leq 0, \nonumber
\end{align}
and hence the origin of \eqref{eq:ht_normalized_tilde} is uniformly globally stable.
Next, we establish uniform global attractivity by building Matrosov functions as follows. Let $V_{1}$ be defined as in \eqref{eq:V1} so that, following the steps leading up to \eqref{eq:Y1}, we may write, for all $(\tilde{x}, t) \in \mathbb{R}^{2n} \times \mathbb{R}_{\geq 0}$,
\begin{align*}
&\frac{\partial V_{1}(\tilde{x}, t)}{\partial t} + \frac{\partial V_{1}(\tilde{x}, t)}{\partial \tilde{x}} f(\tilde{x}, t) \\
&\qquad\leq -\exp(-T) \delta \tilde{\theta}^{T} \tilde{\theta} + \tilde{\theta}^{T}\phi\phi^{T}\tilde{\theta} + \beta M^{2}|\tilde{\theta}||p| \\
&\qquad\eqqcolon Y_{1}(\tilde{x}, \phi).
\end{align*}
Note that $Y_{0}(\tilde{x}, \phi) = 0$ implies $p = 0$ and $\phi^{T}\tilde{\theta} = 0$, which implies $Y_{1}(\tilde{x}, \phi) = -\exp(-T) \delta \tilde{\theta}^{T} \tilde{\theta} \leq 0$. Also note that $Y_{0}(\tilde{x}, \phi) = Y_{1}(\tilde{x}, \phi) = 0$ implies $p = 0$ and $\tilde{\theta} = 0$. Finally, note that $V_{0}$ is time-invariant, the maps $(\tilde{x}, t) \mapsto V_{1}(\tilde{x}, t)$ and $(\tilde{x}, t) \mapsto \phi(t)$ are each locally bounded in $\tilde{x}$ uniformly in $t$, and both $Y_{0}$ and $Y_{1}$ are continuous. Thus, the conditions of Matrosov's theorem \cite{Loria05a} hold, and we conclude that the origin of \eqref{eq:ht_normalized_tilde} is UGAS.
\end{proof}

Theorems \ref{thm:ugas_ht} and \ref{thm:ugas_ht_normalized} also establish uniform global exponential stability (UGES), due to linearity of the systems \eqref{eq:ht} and \eqref{eq:ht_normalized} and the fact that UGAS is equivalent to UGES for linear time-varying systems \cite[Thm. 58.7]{Hahn67a}.

Our assumptions differ from those of \cite{gaudio-2021} only in regards to the regressor’s properties. Namely, the analyses previously reported in \cite[Thm. 2]{gaudio-2021} and \cite[Remark 8]{gaudio-2021} require that the regressor has a bounded time derivative, whereas our analyses do not require differentiability of the regressor but instead require that it be persistently exciting. As a consequence, the previous analyses can establish only that the output error $e_{y} \coloneqq \phi^{T}\tilde{\theta}$ tends to $0$ and not necessarily uniformly, whereas Theorems \ref{thm:ugas_ht} and \ref{thm:ugas_ht_normalized} establish that the parameter error $\tilde{\theta}$ tends to $0$ uniformly.

\section{Concurrent learning for high-order tuners}
\label{sec:CL}

\subsection{Stability analysis}

Let $\left\{(\phi\left(t_{k}), \; y^{*}(t_{k})\right)\right\}_{k=1}^{N}$ be a sequence of recorded data. Define $B: \mathbb{R}^{n} \times \mathbb{R}_{\geq 0} \rightarrow \mathbb{R}^{n}$ as
\begin{align}
B(\theta, \mu) \coloneqq \sum_{k=1}^{N} \frac{\phi\left(t_{k}\right)}{1 + \mu\phi^{T}\left(t_{k}\right)\phi\left(t_{k}\right)} \left(\phi^{T}\left(t_{k}\right) \theta - y^{*}\left(t_{k}\right)\right). \label{eq:B}
\end{align}
In \eqref{eq:ht}, we implement concurrent learning (CL) in the sense of, e.g., \cite{chowdhary-2010},  \cite{chowdhary-2013}, and \cite{ochoa-2021}, as follows:
\begin{align}
\label{eq:ht_CL}
\dot{x} &= f(x, t) \coloneqq \left[\begin{array}{c}-\beta(\theta - \vartheta)\mathcal{N}_{t} \\ - \gamma \left(\nabla_{\theta} L_{t}(\theta) + \mathcal{N}_{t}B(\theta, \mu)\right)\end{array}\right],
\end{align}
where $\mathcal{N}_{t}$ and $\mu$ are given by \eqref{eq:Nt}, and $\beta$ and $\gamma$ are positive real numbers. The $B$-term involves a factor of $\mathcal{N}_{t}$ for reasons that will become clear in the stability analysis. 

In \eqref{eq:ht_normalized}, we implement CL as follows:
\begin{align}
\label{eq:ht_normalized_CL}
\dot{x} &= f(x, t) \coloneqq \left[\begin{array}{c}-\beta(\theta - \vartheta) \\ - \gamma \left(\frac{1}{\mathcal{N}_{t}} \nabla_{\theta} L_{t}(\theta) + B(\theta, \mu)\right)\end{array}\right],
\end{align}
where $\mathcal{N}_{t}$ and $\mu$ are given by \eqref{eq:Nt}, and $\beta$ and $\gamma$ are positive real numbers.

UGAS properties for \eqref{eq:ht_CL} and \eqref{eq:ht_normalized_CL} can be shown if the data satisfies the following property, which is characterized by the subsequent lemma.
\begin{assumption}
\label{assumption_SR}
The regressor data $\left\{\phi\left(t_{k}\right)\right\}_{k=1}^{N}$ is sufficiently rich (SR) in the sense that the matrix $${\mathcal{D} \coloneqq \left[\phi\left(t_{1}\right), \; \phi\left(t_{2}\right), \; \cdots, \;  \phi\left(t_{N}\right)\right] \in \mathbb{R}^{n \times N}}$$ has rank $n$.
\end{assumption}
\begin{lemma}
\label{thm:SR}
For a given $\mu \in \Re_{\geq 0}$, Assumption \ref{assumption_SR} holds if and only if there exists $\delta_{\mu} \in \mathbb{R}_{>0}$ such that $$P_{\mu} \coloneqq \sum_{k=1}^{N} \frac{\phi\left(t_{k}\right) \phi^{T}\left(t_{k}\right)}{1 + \mu \phi^{T}\left(t_{k}\right) \phi\left(t_{k}\right)} \geq \delta_{\mu} I_{n}.$$
\end{lemma}
\begin{proof}
Let $\mu \in \Re_{\geq 0}$ be given. First, we show the forward implication. Assuming that $\mathcal{D}$ has rank $n$, it follows that, for any non-zero $x \in \Re^{n}$, there exists $k \in \{1, \ldots, N\}$ such that $\phi^{T}(t_{k})x \neq 0$. (If it were not true, the columns of $\mathcal{D}$ would not span $\Re^{n}$.) In other words, for any non-zero $x \in \Re^{n}$, there exists $k$ such that $x^{T}\phi(t_{k})\phi^{T}(t_{k})x > 0$. It follows that, for any non-zero $x \in \Re^{n}$, $x^{T}P_{\mu}x > 0$, and hence there exists $\delta_{\mu} \in \Re_{>0}$ (which generally depends on $\mu$) such that $x^{T}P_{\mu}x \geq \delta_{\mu}$. Next, we show the reverse implication. Assuming that there exists $\delta_{\mu} \in \Re_{>0}$ such that $P_{\mu} \geq \delta_{\mu} I_{n}$, it follows that, for any non-zero $x \in \Re^{n}$, there exists $k \in \{1, \ldots, N\}$ such that $x^{T}\phi(t_{k})\phi^{T}(t_{k})x > 0$. That is, there exists $k$ such that $\phi^{T}(t_{k})x \neq 0$. Then, because $x$ is an arbitrary non-zero vector in $\Re^{n}$, it follows that there are $n$ linearly independent columns of $\mathcal{D}$.
\end{proof}

\begin{theorem}
\label{thm:ugas_ht_CL}
Under Assumption \ref{assumption_SR}, if $\beta \geq 2\gamma/\mu$, the point $(\theta^{*}, \theta^{*})$ is UGAS for \eqref{eq:ht_CL}.
\end{theorem}
\begin{proof}
Let $\tilde{\theta} \coloneqq \theta - \theta^{*}$, $p \coloneqq \vartheta - \theta$, and $\tilde{x} \coloneqq (\tilde{\theta}, p)$. Due to the fact that $y^{*}\left(t_{k}\right) = \phi^{T}\left(t_{k}\right)\theta^{*}$, we have that $B(\theta, \mu) = P_{\mu}\tilde{\theta}$, and (reusing notation from previous proofs) \eqref{eq:ht_CL} can be written as
\begin{align}
\label{eq:ht_tilde_CL}
\dot{\tilde{x}} &= \tilde{f}(\tilde{x}, t) \coloneqq \left[\!\!\begin{array}{c}\beta\mathcal{N}_{t}p  \\ -\beta\mathcal{N}_{t}p - \gamma \left(\nabla_{\theta} L_{t}(\tilde{\theta} + \theta^{*}) + \mathcal{N}_{t}P_{\mu}\tilde{\theta}\right)\end{array}\!\!\right].
\end{align}
To show that $(\theta^{*}, \theta^{*})$ is UGAS for \eqref{eq:ht_CL}, it suffices to show that the origin of \eqref{eq:ht_tilde_CL} is UGAS. Consider the Lyapunov function candidate
\begin{align}
V(\tilde{x}) &\coloneqq \frac{1}{\gamma}\left|\tilde{\theta} + p\right|^{2} + \frac{1}{\gamma}|p|^{2} + \frac{2}{\beta}\tilde{\theta}^{T}P_{\mu}\tilde{\theta}, \label{eq:V_CL}
\end{align}
which is radially unbounded, positive definite, and continuously differentiable. Following steps similar to those leading up to \eqref{eq:Y0}, it can be shown that, for all $(\tilde{x}, t) \in \mathbb{R}^{2n} \times \mathbb{R}_{\geq 0}$ and for $\beta \geq 2\gamma/\mu$,
\begin{align}
&\left\langle\nabla V(\tilde{x}), \tilde{f}(\tilde{x}, t)\right\rangle \nonumber \\ 
&\leq -\frac{2\gamma}{\beta}|p|^{2} - |e_{y}|^{2} - 2\mathcal{N}_{t}\left\langle \tilde{\theta}, P_{\mu}\tilde{\theta}\right\rangle - 4\mathcal{N}_{t}\left\langle p, P_{\mu}\tilde{\theta}\right\rangle \nonumber \\
&\qquad + \mathcal{N}_{t}\left\langle\frac{4}{\beta}P_{\mu}\tilde{\theta}, \; \beta p\right\rangle \nonumber \\
&\leq -2\mathcal{N}_{t}\tilde{\theta}^{T}P_{\mu}\tilde{\theta} - \frac{2\beta}{\gamma}|p|^{2} \label{eq:Y_CL} \\
&\leq -2\tilde{\theta}^{T}P_{\mu}\tilde{\theta} - \frac{2\beta}{\gamma}|p|^{2} \eqqcolon Y(\tilde{x}). \nonumber
\end{align}
Due to Lemma \ref{thm:SR}, $Y$ is negative definite. Hence, $V$ is a Lyapunov function for \eqref{eq:ht_tilde_CL}, and the origin of \eqref{eq:ht_tilde_CL} is UGAS.
\end{proof}

\begin{theorem}
\label{thm:ugas_ht_normalized_CL}
Under Assumption \ref{assumption_SR}, with $\mathcal{N}_{t}$ given by \eqref{eq:Nt}, if $\beta \geq 2\gamma/\mu$, the point $(\theta^{*}, \theta^{*})$ is UGAS for \eqref{eq:ht_normalized_CL}.
\end{theorem}
\begin{proof}
Let $\tilde{\theta} \coloneqq \theta - \theta^{*}$, $p \coloneqq \vartheta - \theta$, and $\tilde{x} \coloneqq (\tilde{\theta}, p)$. Due to the fact that $y^{*}\left(t_{k}\right) = \phi^{T}\left(t_{k}\right)\theta^{*}$, we have that $B(\theta, \mu) = P_{\mu}\tilde{\theta}$ for any $\mu \in \Re_{\geq 0}$, and (reusing notation from previous proofs) \eqref{eq:ht_normalized_CL} can be written as
\begin{align}
\label{eq:ht_normalized_tilde_CL}
\dot{\tilde{x}} &= \tilde{f}(\tilde{x}, t) \coloneqq \left[\begin{array}{c}\beta p \\ -\beta p - \gamma \left(\frac{1}{\mathcal{N}_{t}} \nabla_{\theta} L_{t}(\tilde{\theta} + \theta^{*}) + P_{\mu}\tilde{\theta}\right)\end{array}\right].
\end{align}
To show that $(\theta^{*}, \theta^{*})$ is UGAS for \eqref{eq:ht_normalized_CL}, it suffices to show that the origin of \eqref{eq:ht_normalized_tilde_CL} is UGAS. Consider the Lyapunov function candidate \eqref{eq:V_CL}, which is radially unbounded, positive definite, and continuously differentiable. Observing that the right-hand side of \eqref{eq:ht_normalized_tilde_CL} can be obtained by multiplying the right-hand side of \eqref{eq:ht_tilde_CL} by $1/\mathcal{N}_{t}$, we have from \eqref{eq:Y_CL} that, for all $(\tilde{x}, t) \in \mathbb{R}^{2n} \times \mathbb{R}_{\geq 0}$ and for $\beta \geq 2\gamma/\mu$,
\begin{align*}
&\left\langle\nabla V(\tilde{x}), \tilde{f}(\tilde{x}, t)\right\rangle \\
&\qquad\leq -2\tilde{\theta}^{T}P_{\mu}\tilde{\theta} - \frac{2\beta}{\gamma\mathcal{N}_{t}}|p|^{2} \\
&\qquad\leq -2\tilde{\theta}^{T}P_{\mu}\tilde{\theta} - \frac{2\beta}{\gamma(1+\mu M^{2})}|p|^{2} \eqqcolon Y(\tilde{x}).
\end{align*}
Due to Lemma \ref{thm:SR}, $Y$ is negative definite. Hence, $V$ is a Lyapunov function for \eqref{eq:ht_normalized_tilde_CL}, and the origin of \eqref{eq:ht_normalized_tilde_CL} is UGAS.
\end{proof}

One benefit of the CL feature is that it can be implemented using only the sampled data, mitigating the practical costs that may be incurred by continuously collecting data online. For example, the following system can be interpreted as a novel high-order tuner which does not require continuous online measurement of $\phi(t)$ and $y(t)$:
\begin{align}
\label{eq:ht_B}
\dot{x} &= f(x, t) \coloneqq \left[\begin{array}{c}-\beta(\theta - \vartheta) \\ - \gamma B(\theta, \mu)\end{array}\right].
\end{align}
In this case, an additional advantage is that the UGAS property can be established under no restriction on the relative magnitudes of $\beta$, $\gamma$, and $\mu$.
\begin{theorem}
\label{thm:ugas_ht_B}
Under Assumption \ref{assumption_SR}, for any positive $\beta$ and $\gamma$, and for any nonnegative $\mu$, the point $(\theta^{*}, \theta^{*})$ is UGAS for \eqref{eq:ht_B}.
\end{theorem}
\begin{proof}
Let $\tilde{\theta} \coloneqq \theta - \theta^{*}$, $p \coloneqq \vartheta - \theta$, and $\tilde{x} \coloneqq (\tilde{\theta}, p)$. Due to the fact that $y^{*}\left(t_{k}\right) = \phi^{T}\left(t_{k}\right)\theta^{*}$, we have that $B(\theta, \mu) = P_{\mu}\tilde{\theta}$ for any $\mu \in \Re_{\geq 0}$, and (reusing notation from previous proofs) \eqref{eq:ht_B} can be written as
\begin{align}
\label{eq:ht_normalized_tilde_B}
\dot{\tilde{x}} &= \tilde{f}(\tilde{x}, t) \coloneqq \left[\begin{array}{c}\beta p \\ -\beta p - \gamma P_{\mu}\tilde{\theta}\end{array}\right].
\end{align}
To show that $(\theta^{*}, \theta^{*})$ is UGAS for \eqref{eq:ht_B}, it suffices to show that the origin of \eqref{eq:ht_normalized_tilde_B} is UGAS. Consider the Lyapunov function candidate
\begin{align}
V(\tilde{x}) &\coloneqq \frac{1}{2}|\tilde{\theta} + p|^{2} + \frac{1}{2}|p|^{2} + \frac{\gamma}{\beta}\tilde{\theta}^{T}P_{\mu}\tilde{\theta}, \nonumber
\end{align}
which is radially unbounded, positive definite, and continuously differentiable. For all $\tilde{x} \in \Re^{2n}$, we have
\begin{align*}
&\left\langle\nabla V(\tilde{x}), \tilde{f}(\tilde{x}, t)\right\rangle = \left\langle\tilde{\theta} + p + \frac{2\gamma}{\beta}P_{\mu}\tilde{\theta}, \beta p\right\rangle \\
&\qquad\qquad- \langle\tilde{\theta} + 2p, \beta p + \gamma P_{\mu}\tilde{\theta}\rangle \\
&\qquad= \langle p, 2\gamma P_{\mu}\tilde{\theta}\rangle - \gamma\tilde{\theta}^{T}P_{\mu}\tilde{\theta} - \beta|p|^{2} - \langle 2p, \gamma P_{\mu}\tilde{\theta}\rangle \\
&\qquad= -\gamma\tilde{\theta}^{T}P_{\mu}\tilde{\theta} - \beta|p|^{2} \eqqcolon Y(\tilde{x}).
\end{align*}
Due to Lemma \ref{thm:SR}, $Y$ is negative definite. Hence, $V$ is a Lyapunov function for \eqref{eq:ht_normalized_tilde_B}, and the origin of \eqref{eq:ht_normalized_tilde_B} is UGAS.
\end{proof}

As remarked in Section \ref{sec:PE}, UGAS is equivalent to UGES for linear time-varying systems, and therefore Theorems \ref{thm:ugas_ht_CL}, \ref{thm:ugas_ht_normalized_CL}, and \ref{thm:ugas_ht_B} also establish UGES, due to linearity of the systems \eqref{eq:ht_CL}, \eqref{eq:ht_normalized_CL}, and \eqref{eq:ht_B}.

\subsection{Online implementation}
\label{sec:online_CL}

Standard analyses of stability and convergence for CL methods assume the availability of recorded (offline) data that satisfies a rank condition similar to Assumption \ref{assumption_SR} \cite{chowdhary-thesis}, \cite{chowdhary-2010}, and the same assumption is made in Theorems \ref{thm:ugas_ht}-\ref{thm:ugas_ht_B}. In the absence of such data, CL can be implemented via online criteria that aim to ensure that the rank condition is satisfied after executing the algorithm for some period of time. When evaluating the performance of high-order tuners in Section \ref{sec:numerical}, the simple online criterion of \cite[Sec. IV]{chowdhary-2011} is considered. Namely, the number of recorded data points at time $t$ is denoted $N(t)$, and a new data point is added to the data set $\left\{(\phi\left(t_{k}), \; y^{*}(t_{k})\right)\right\}_{k=1}^{N(t)}$ at time $t$ if the following condition is satisfied for some user-specified parameter $\varepsilon \in \Re_{>0}$:
\begin{align}
    \frac{|\phi(t) - \phi(t_{N(t)})|^{2}}{|\phi(t)|} \geq \varepsilon. \label{eq:online_criterion}
\end{align}
The data set is initialized as $\left\{(\phi\left(t_{1}), \; y^{*}(t_{1})\right)\right\}$ with $t_{1} = 0$ and $N(0) = 1$. The condition \eqref{eq:online_criterion} is evaluated until $N(t)$ reaches a user-specified maximum, an integer $\overline{N} \geq n$. Equation \eqref{eq:B} is implemented in a time-dependent fashion, i.e., with $N$ replaced by $N(t)$, and it becomes time-independent after $N(t)$ reaches $\overline{N}$. 

For this online implementation of CL, the stability analysis is beyond the scope of the current work and, to our knowledge, has not been pursued in previous studies of CL methods. Previous studies of online implementations have focused not on stability analysis of the resulting dynamics but on showing that, in certain circumstances, an online criterion ensures that a rank-condition is satisfied in finite time or that the convergence rate of the resulting dynamics is maximized after a rank-condition has been met \cite{chowdhary-2011}.

\section{Soft-reset methods for high-order tuners}
\label{sec:softreset}

For applications that demand high efficiency and precision, we explore the possibility that resetting the state $\vartheta - \theta$ to zero under certain conditions can benefit performance, motivated by works such as \cite{le-2021} and \cite{le-2021b}. Defining $f$ as in \eqref{eq:ht_CL}, we incorporate resets in \eqref{eq:ht_CL} via a soft-reset approach, similar to ideas proposed in \cite{le-2021b} and \cite{teel-2021}, resulting in a differential inclusion given by
\begin{subequations}
\label{eq:ht_CL_softreset}
\begin{align}
&\varphi(x, t) \coloneqq \left\langle \vartheta - \theta, \; \nabla_{\theta} L_{t}\left(\theta\right)\right\rangle, \\
&\dot{x} \in f(x, t) + \beta_{r}\left(\vphantom{\frac{}{}} \mbox{\textrm SGN}(\varphi(x, t)) + 1 \right)  \left[\begin{array}{c}-(\theta - \vartheta)\mathcal{N}_{t} \\ 0\end{array}\right],
\end{align}
\end{subequations}
where $\beta_{r} \in \Re_{>0}$. The mapping $\mbox{\rm `SGN'}$ is the set-valued sign mapping, i.e., $\mbox{\rm SGN}(s)$ is equal to $s/|s|$ when $s \neq 0$ and $\mbox{\rm SGN}(0)=[-1,1]$. Consequently, when $\varphi(x, t) < 0$, \eqref{eq:ht_CL_softreset} behaves like \eqref{eq:ht_CL}. On the other hand, when $\varphi(x, t) > 0$, the reset behavior becomes active, causing \eqref{eq:ht_CL_softreset} to behave like \eqref{eq:ht_CL_softreset} but with $\beta$ increased to a value of $\beta + 2\beta_{r}$. Based on these observations, a UGAS result for \eqref{eq:ht_CL_softreset} can be obtained under the same conditions of Theorem \ref{thm:ugas_ht_CL}, following similar steps in the proof and making use of Lyapunov conditions for differential inclusions.

Defining $f$ as in \eqref{eq:ht_normalized_CL}, soft resets are implemented in \eqref{eq:ht_normalized_CL} according to the following differential inclusion:
\begin{subequations}
\label{eq:ht_normalized_CL_softreset}
\begin{align}
&\varphi(x, t) \coloneqq \left\langle \vartheta - \theta, \; \frac{1}{\mathcal{N}_{t}}\nabla_{\theta} L_{t}\left(\theta\right)\right\rangle, \\
\dot{x} &\in f(x, t) + \beta_{r}\left( \vphantom{\frac{}{}}  \mbox{\textrm SGN}(\varphi(x, t)) + 1 \right)  \left[\begin{array}{c}-(\theta - \vartheta) \\ 0\end{array}\right],
\end{align}
\end{subequations}
with $\beta_{r} \in \Re_{>0}$. A UGAS result for \eqref{eq:ht_normalized_CL_softreset} can be obtained using similar observations as those made above for \eqref{eq:ht_CL_softreset}, following similar steps as in the proof of Theorem \ref{thm:ugas_ht_normalized_CL} and making use of Lyapunov conditions for differential inclusions.

In contrast with the soft-reset approach, the authors of \cite{ochoa-2021} have studied various hard-reset approaches in high-order algorithms for parameter identification.

\section{Numerical results}
\label{sec:numerical}


\begin{figure}
    \centering
    \includegraphics[width=0.6\linewidth]{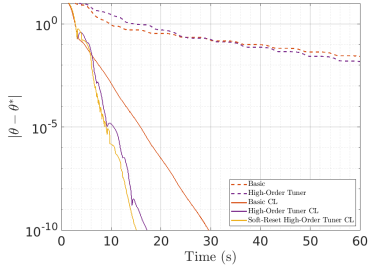}
    \caption{Evolution of the parameter error norm over time, for the algorithms having an unnormalized gradient term.}
    \label{fig:CL}
\end{figure}

Figure~\ref{fig:CL} compares the efficiency of the basic gradient method given by
\begin{align}
\label{eq:basic}
\dot{\theta} &= -\nabla_{\theta} L_{t}(\theta)
\end{align}
and of the high-order tuner given by \eqref{eq:ht}, along with their respective CL counterparts given by
\begin{align}
\label{eq:basic_CL}
\dot{\theta} &= -\gamma \left(\nabla_{\theta} L_{t}(\theta) + B(\theta, 0)\right)
\end{align}
and equation \eqref{eq:ht_CL}, using the online implementation described in Section \ref{sec:online_CL} with $\varepsilon = 1$ and $\overline{N} = 10$. The soft-reset system is given by \eqref{eq:ht_CL_softreset}. Following \cite[Sec. 5.7.1]{gaudio-dissertation}, we set $\phi(t) = [1, \; 1 + 3\sin(t), \; 1 + 3\cos(t)]^{T}$ for all $t \geq 0$ and randomly initialize the value of $\theta$. We initialize $\vartheta$ at the same value as $\theta$. For all algorithms, $\gamma = 0.1$, and $\mu = 0.2$. For all high-order tuners, we set $\beta = 1$ to satisfy the requirements of Theorems \ref{thm:ugas_ht} and \ref{thm:ugas_ht_CL}. These parameter values are intended to match those chosen in \cite[Sec. 5.7]{gaudio-dissertation}. For the soft-reset system, $\beta_{r} = 4$. All systems are implemented by Euler discretization with a stepsize of $10^{-3}$.

\begin{figure}
    \centering
    \includegraphics[width=0.6\linewidth]{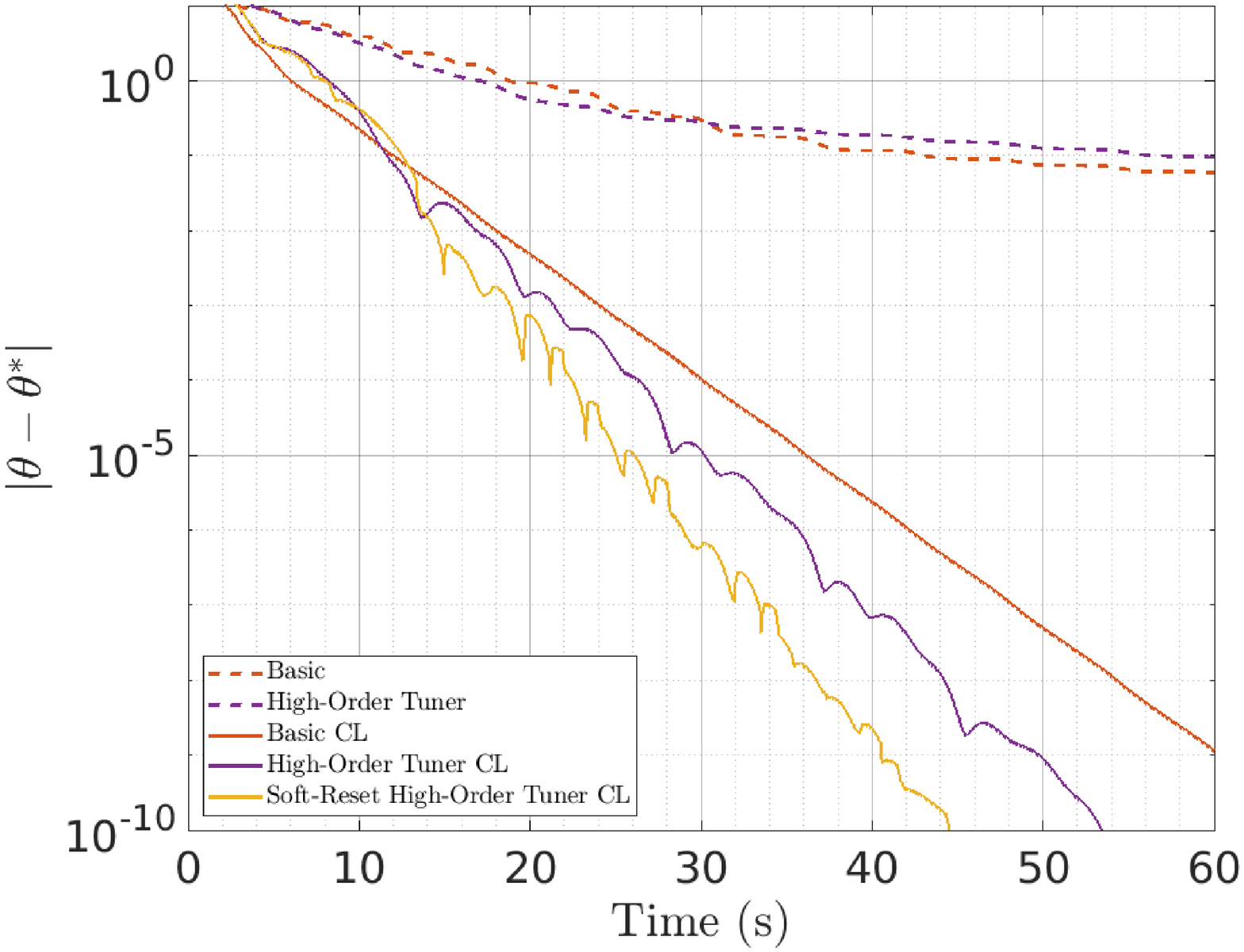}
    \caption{Evolution of the parameter error norm over time, for the algorithms having a normalized gradient term.}
    \label{fig:CL_normalized}
\end{figure}

Under the same conditions as in Figure \ref{fig:CL}, Figure \ref{fig:CL_normalized} compares the efficiency of the basic normalized gradient method given by
\begin{align}
\label{eq:basic_normalized}
\dot{\theta} &= -\frac{\gamma}{\mathcal{N}_{t}} \nabla_{\theta} L_{t}(\theta)
\end{align}
and of the high-order tuner given by \eqref{eq:ht_normalized}, along with their respective CL counterparts given by
\begin{align}
\label{eq:basic_normalized_CL}
\dot{\theta} &= -\gamma \left(\frac{1}{\mathcal{N}_{t}} \nabla_{\theta} L_{t}(\theta) + B(\theta, \mu)\right)
\end{align}
and equation \eqref{eq:ht_normalized_CL}. To satisfy the requirements of Theorems \ref{thm:ugas_ht_normalized} and \ref{thm:ugas_ht_normalized_CL}, the values of $\beta$, $\gamma$, $\mu$, and $\beta_{r}$ are chosen to be the same as in the experiment shown by Figure \ref{fig:CL}.

\printbibliography

\end{document}